\documentclass[runningheads]{llncs}
\usepackage{soul}
\usepackage{url}
\usepackage{graphicx}
\usepackage{amsmath}
\usepackage{amssymb}
\usepackage{booktabs}
\usepackage{algorithm}
\usepackage{algorithmic}
\usepackage{tikz} 
\usepackage{dsfont}
\usepackage{bm}
\usepackage{xfrac}
\usepackage{dsfont}

\usepackage{subcaption}

\usepackage{enumitem}
\usepackage{xcolor}
\usepackage{tikzscale}
\usetikzlibrary{arrows,decorations.pathmorphing,backgrounds,positioning,fit,petri}

\newcommand{\floor}[1] {\lfloor #1\rfloor}
\usepackage{xspace}
\newcommand{\peer}{\text{PeerNomination}\xspace}
\newcommand{\extpeer}{\text{\textsc{WeightedPeerNomination}}\xspace}

\begin{document}
\title{Peer Selection with Noisy Assessments}

\author{Omer Lev\inst{1} \and Nicholas Mattei\inst{2} \and Paolo Turrini\inst{3} \and 
Stanislav Zhydkov \inst{3}}
\authorrunning{O. Lev, N. Mattei, P. Turrini and S. Zhydkov}

\institute{Ben-Gurion University, Israel \\\email{omerlev@bgu.ac.il} 
\and 
Tulane University, USA \\\email{nsmattei@tulane.edu}
\and 
University of Warwick, UK \\\email{\{p.turrini, s.zhydkov\}@warwick.ac.uk}
}
\maketitle              %

\begin{abstract}

In the peer selection problem a group of agents must select a subset of themselves as winners for, e.g., peer-reviewed grants or prizes. Here, we take a Condorcet view of this aggregation problem, i.e., that there is a ground-truth ordering over the agents and we wish to select the best set of agents, subject to the noisy assessments of the peers. Given this model, some agents may be unreliable, while others might be self-interested, attempting to influence the outcome in their favour. In this paper we extend \peer, the most accurate peer reviewing algorithm to date, into \extpeer, which is able to handle noisy and inaccurate agents. To do this, we explicitly formulate assessors' reliability weights in a way that does not violate strategyproofness, and use this information to reweight their scores. We show analytically that a weighting scheme can improve the overall accuracy of the selection significantly. Finally, we implement several instances of reweighting methods and show empirically that our methods are robust in the face of noisy assessments.
\end{abstract}

\section{Introduction}\label{sec:intro}

The results of the 2016 NeurIPS experiment \cite{shah2018design} and other studies of bias in evaluative processes \cite{wang2020debiasing,stelmakh2020catch} have brought to the fore the extent of noisy assessments in conference reviewing. When agents evaluate each other to select a subset of themselves -- the peer selection problem -- various factors can come into play that hinder accurate assessment, including time pressure and strategic behaviour. Finding high quality mechanisms for peer review is a critical step in helping to ease reviewing in a number of areas, including large conferences \cite{Aziz2019}, grant reviewing \cite{MerrifieldSaari}, and online courses \cite{Walsh14}.

Recently, researchers in algorithmic game theory and computational social choice have devised accurate and strategyproof algorithms for peer selection, such as Partition \cite{AFPT11a}, Credible Subset \cite{KLMP15}, Exact Dollar Partition (EDP) \cite{Aziz2019} and most recently \peer \cite{PeerNomination}, each improving upon the state of the art. These algorithms take a Condorcet view on this aggregation problem, i.e., that there is a ground-truth over the agents, and we wish to select as many of the top ranked agents as possible \cite{Xia19}. However, none of the existing algorithms seek to alleviate the problem of noisy inputs in a unified, strategyproof mechanism. When earlier work did engage with noisy reports, it was limited to empirical testing with relatively low noise, e.g., a Mallows model with $\varphi=0.5$ (see  \cite{Aziz2019} or \cite{PeerNomination}), which yields fairly minor changes in agents' reports. We are concerned with algorithms that can handle a significant level of noise, while maintaining strategyproofness and high quality of selection, an important missing aspect in the peer reviewing literature.

Ideally, we want an algorithm that is capable of identifying inaccurate reviewers and reducing their influence on selection, using their own assessments as a guide. We could, for example, try and downgrade those reviewers that differ too much from others. However, there are two problems with this approach: first, the noise may be such that it is difficult to establish what the consensus actually is; and second, that this meta-level reweghting can be exploited strategically. Simple reweghting is not strategyproof. Consider an agent $a$ that is harshly reviewing agent $b$, with both $a$ and $b$ reviewing a third agent $c$. Agent $b$ could benefit by reviewing agent $c$ in a way that would present agent $a$ as one whom others disagree with, potentially lowering the impact of the report of agent $a$ on agent $b$ if weights are computed based on correlations to the evaluations of others, as done by \cite{MerrifieldSaari}. On the other hand, if a mechanism is able to identify agent $b$ as a source of noise, it can increase the overall quality of the selection. While one can reweigh agents without maintaining strategyproofness \cite{Walsh14,NIPS2009_f899139d}, we wish to achieve increased selection quality \emph{and} strategyproofness.

\smallskip
\textbf{Contribution.} We extend \peer, the most accurate peer reviewing algorithm to date, into \extpeer, to handle noisy and inaccurate agents. To do so we explicitly formulate assessors' reliability weights in a way that does not violate strategyproofness, and use this information to reweigh their scores.  \extpeer is able to handle significant levels of noise, even when reviewers act adversarially. We show analytically that a weighting scheme can improve the overall accuracy of the selection significantly. We then implement several instances of reweghting protocols and show empirically that our novel methods are able to significantly improve the quality of peer selection over \peer, under a variety of noise parameters.

\section{Related Literature}

Using the evaluations of peers to rank and select winners is a problem of broad interest beyond CS and AI, including numerous practical domains, e.g., conference, journal, and grant reviewing; large scale course grading, and group decision making. Brought to the fore by \cite{MerrifieldSaari} to allocate telescope time, the problem is deeply rooted in economics, from the work of \cite{Dollar} on ``dollar partition", extended by \cite{Aziz2019} to the Dollar Raffle and Dollar Partition methods.
Other notable algorithms include the Credible Subset method \cite{KLMP15},
where the protocol examines the possibility of manipulations and accounts for it. Despite strategyproofness, the system was shown to yield significant number of cases where no proposal was funded in the end \cite{Aziz2016}.

Two more prominent recent algorithms are Exact Dollar Partition (EDP) \cite{Aziz2019} which provides exactness at the cost of some randomness, while remaining strategyproof, and improving on earlier algorithms \cite{Aziz2019}. The second is \peer \cite{PeerNomination}, which improves upon EDP, while requiring reviewers to submit only approval based rankings, though at the cost of some exactness (details in Section~\ref{prem}).

Other developments in the multi-agent systems communities include voting rules to aggregate ranks, e.g., $k$-Partite  \cite{KKKKP18a}, the Committee Rule \cite{KKKKP18a}, and Divide-and-Rank \cite{XZSS19} algorithms. Others focus on proving bounds on the quality of a given rank aggregation scheme under noisy and partial observations \cite{DBLP:conf/atal/CaragiannisKV15}. Yet other methods are approval-based but focus on single agent selection: Permutation~\cite{FeKl14a} and Slicing~\cite{BNV14a}.

A key application area for peer evaluation mechanisms is education, where the problems of reviewer reliability and bias have been extensively studied \cite{DBLP:conf/edm/PiechHCDNK13}. We are motivated by evidence from fielded peer evaluation mechanisms showing that students are often unwilling to strictly rank assignments \cite{DBLP:conf/sigcse/AlfaroS14} and would rather rely on scores or pass/fail marks (approvals). Within the conference and journal reviewing ecosystem there is also growing interest in detecting strategic behaviour on the part of the reviewers \cite{stelmakh2020catch,Meir:MarkedBidding} as well as de-biasing and calibrating differences in the scores of reviewers \cite{wang2019your,LianMNW18}. We go beyond calibration and de-biasing, identifying  suboptimal behaviour in agents' populations and looking at the effect of rescaling on the system as a whole.

Outside peer selection, there is extensive work in the machine learning, information retrieval, and preference learning communities on the \emph{learning to rank} problem: inferring the most likely ranking from possibly noisy observations \cite{liu2011learning}. These works include learning noise models, e.g., the parameters of a Mallows model, for use in inferring latent preferences of agents \cite{liu2011learning,Xia19}. This is of great practical interest in information retrieval, where one wishes to rank, e.g., web-pages based off user clicks \cite{schnabel2016unbiased} and in combining labellings from multiple sources for the construction of datasets \cite{NIPS2009_f899139d}.  However, all of these systems do not concern themselves with strategyproofness, a key focus of our study.

\section{Preliminaries}\label{prem}

In our setup, a set of agents $\mathcal{N} = \{ 1, 2, ..., n \}$ aim to select a subset of themselves of size $k$. We assume that, if agents were to assess each other accurately, they would report the same ranking. We refer to this ranking as the {\em ground truth}, a standard assumption in Condorcet views of voting \cite{Xia19}. In practical applications, it is often not feasible for agents to review all others, therefore we assume, for simplicity, that each agent reviews $m$ agents and is reviewed by $m$ of them. 
We represent such an $m$-regular graph via an assignment function $A: \mathcal{N} \rightarrow 2^{\mathcal{N}}$ and denote $i$'s review pool, the agents reviewed by $i$, as $A(i)$. Let $A^{-1}(i)$ denote the set of agents that review agent $i$. In real-world settings, $m$ is typically a small constant w.r.t. $n$.

Since we assume the ground truth to be a linear order, we also assume the belief of each reviewer over their review pool to be a strict ranking. Formally, we represent the underlying ranking of reviewer $i$ as an injective function $\sigma_i: A(i) \rightarrow \{1, ..., m \}$, where $\sigma_i(j)$ is the rank given by reviewer $i$ to agent $j$. The collection of these reviews is called a review \textit{profile} and is denoted by $\sigma$. However, as in \cite{PeerNomination}, we take a less demanding approach and require reviewers to only report \emph{approvals} to the mechanism, rather than a strict ranking. Formally, we instantiate each underlying belief $\sigma_i$ by $\sigma_i^\textrm{app}: A(i) \rightarrow \{0, 1, \alpha\}$, where $1$ represents approval and $\alpha \in [0, 1)$ is a constant representing ``partial" approval. Similarly, $\sigma_i^\textrm{app}(j)$ is the approval score given by reviewer $i$ to reviewer $j$. 

\paragraph{\peer.}
In \peer, each agent, independently and concurrently, {\em nominates for certain} $\lfloor \frac{k}{n} m \rfloor$ of their reviewees, i.e., setting $\sigma_i^\textrm{app}(j) = 1$. In addition, they {\em probabilistically nominate} exactly one other agent, with probability $\frac{k}{n} m - \lfloor \frac{k}{n} m \rfloor$, i.e., setting $\sigma_i^\textrm{app}(j) = \frac{k}{n} m - \lfloor \frac{k}{n} m \rfloor$. In other words, probabilistic nominations become certain nominations with probability $\frac{k}{n} m - \lfloor \frac{k}{n} m \rfloor$.  These partial nominations are used to maintain the size requirements as discussed by \cite{PeerNomination}. An agent $j$ is then selected iff they are nominated by the majority of their reviewers. Since each agent is considered independently for selection, the algorithm is not guaranteed to return exactly $k$ agents, though it closely approximates $k$ \cite{PeerNomination}.

\subsection{Noise Model} \label{sec:mallows}

To model the inaccuracies in reviewers' assessments, we assume that each agent receives a noisy observation of the ground truth according to a Mallows model \cite{Mal57}. Mallows models have been widely used to compare the performance of peer selection algorithms empirically \cite{PeerNomination,Aziz2019}, but so far only studied for very mild levels of noise, e.g., $\varphi=0.5$. 

The Mallows model is parameterised by a dispersion parameter $\varphi \in [0, 1]$ and a reference linear ranking $R$. Given $R$ and $\varphi$, the model induces a probability distribution over all permutations of $R$ such that the probability of the linear order $R^\prime$ is $\pi(R^\prime) \propto \varphi^{KT(R, R^\prime)} $, where $KT(R, R^\prime)$ is the Kendall-$\tau$ distance between $R$ and $R'$. The Kendall-$\tau$ distance counts the number of pairwise disagreements between two rankings. In other words, the probability to find an additional pairwise disagreement from the reference ranking decreases exponentially. Note that, as we vary the dispersion parameter $\varphi$ from $0$ to $1$, the probability distribution over all linear rankings moves from being concentrated on $R$ to being uniform overall possible rankings. 

In our simulations (Section \ref{sec:experiments}), we take the ground truth as the reference ranking and sample a noisy ranking for each agent using the $\varphi$ specified. An important feature of Mallows model is that it can be sampled efficiently \cite{lu2011learning,Xia19}, which allows us to generate a unique reviewer profile for each experiment.
 
In addition, we test our weighting schemes in settings where some reviewers are not just random, but are actively adversarial. We thus extend the range of the dispersion parameter $\varphi$ to $[0, 2]$, where $\varphi^\prime \in (1, 2]$ means that the ranking is sampled using the \textit{inverse} of the ground truth as the reference ranking and $\varphi = 2 - \varphi^\prime$. Thus, the distribution moves smoothly from being concentrated at the ground truth to the inverse ground truth while still being uniform around $1$.

It is important to note that the Mallows model does not produce errors in reviews proportionally to $\varphi$. Suppose that agents have to nominate top 3 out of 9 individuals and their beliefs are given by the Mallows model. With $\varphi=0.5$ only a small number of agents are even going to commit 1 error and we need to increase $\varphi$ to around 0.95 to have any meaningful probability of getting 2/3 nominations wrong. Only under our adversarial extension we start seeing agents that get all 3 nominations wrong. Figure \ref{fig:fp_errors} illustrates this relationship of errors and $\varphi$.

\begin{figure}
    \centering
    \includegraphics[width=0.7\linewidth]{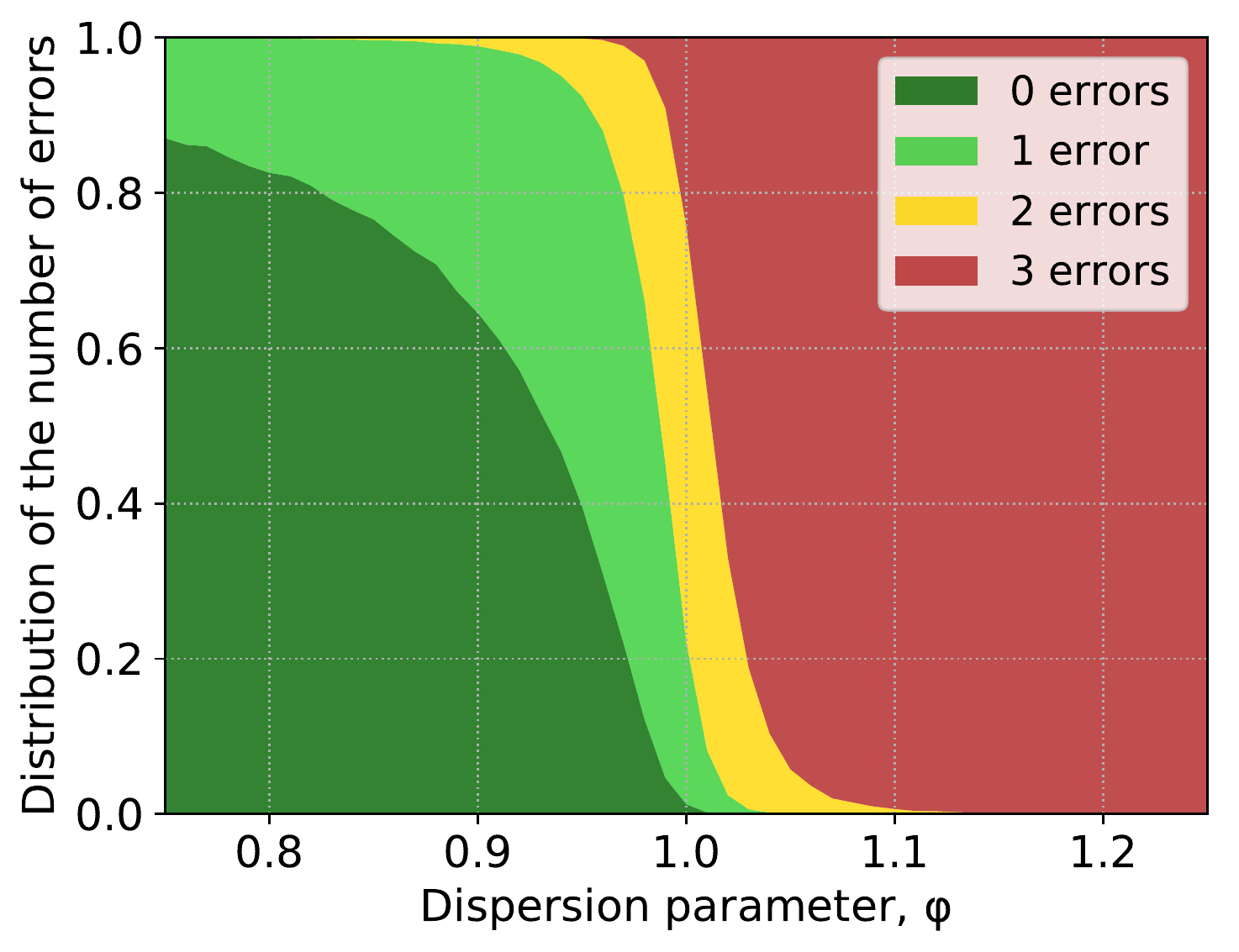}
    \caption{Nomination errors committed by agents as a function of $\varphi$. Each vertical slice shows the the distribution of the number of errors among the 200 agents in a simulation where each agent had to nominate 3 individuals out of 9.}
    \label{fig:fp_errors}
\end{figure}

\subsection{Properties}

Our mechanism, like most, maintains the properties of \emph{anonymity}, permuting agents makes no difference; \emph{non-imposition}, any set of $k$ accepted papers is a possible output; and \emph{monotonicity}, if agent $j$ was selected, and some agent $i$ increased their score for it, agent $j$ will still be selected.

We focus on \emph{strategyproofness}, i.e., no agent is better off by reporting a non-truthful review. We do so by showing no agent's reviews have any influence over their own selection, and hence there is no incentive to misreport; and on ensuring the algorithm's quality. We measure the quality by two key parameters: {\em recall}, the proportion of true positives in the selected set\footnote{Accuracy, which is often paired with recall, is less relevant in our setting as we have a fixed and often small proportion of positives (which is equal to $k/n$) in the population.}; and {\em size}, i.e., the number of selected agents.

\section{Weighted PeerNomination}\label{sec:model}

\extpeer extends \peer \cite{PeerNomination} by adjusting a reviewer's weight as a function of the overall quality of their reviews. This design choice is based on the intuition that in many cases, agents' reviews are not independent when it comes to quality but, rather, correlated. \peer is a special case of \extpeer with weights $(1,1,\ldots,1)$.

The new algorithm can be separated into three parts:
\begin{description}
\item[Assignment.] Determining which agents review each agent (Algorithm~\ref{alg:eulerAlg}).
\item[Weights.] Applying a reweghting protocol to adjust those reviews (Section~\ref{sec:weightSchemes}).
\item [Choice.] Choosing, based on reweighed reviews, which agents to select (Algorithm~\ref{alg:extpeer}).
\end{description}

\begin{algorithm}[t!]
\begin{algorithmic}{}
\small
\REQUIRE Assignment $A$, review profile $\sigma$, target quota $k$, slack parameter $\varepsilon$, reviewer weights $\{w_1, ..., w_n \}$
\ENSURE Accepting set $S$
\STATE Set $\textit{nomQuota} := \frac{k}{n} m +\varepsilon$
\FORALL{$j$ in $\mathcal{N}$}
 \STATE Initialise \textit{nomCount} := 0
 \FORALL{$i \in A^{-1}(j)$}
 \IF{$\sigma_i(j) \leq \lfloor \textit{nomQuota} \rfloor$}
 \STATE increment \textit{nomCount} by $w_i$
 \ELSIF{$\sigma_i(j) = \lfloor \textit{nomQuota} \rfloor + 1$}
 \STATE increment \textit{nomCount} by $w_i$ with probability $\textit{nomQuota} - \lfloor \textit{nomQuota} \rfloor$
 \ENDIF
 \ENDFOR
 \IF{$\textit{nomCount} \geq (\sum_{i \in A^{-1}(j)} w_i)/2$}
 \STATE $S \leftarrow j$
 \ENDIF
\ENDFOR
\RETURN $S$
\end{algorithmic}
\caption{\extpeer \textsc{Choice}}
\label{alg:extpeer}
\end{algorithm}{}

As illustrated in Figure~\ref{fig:sp_diagram} (Left), the weighting scheme introduces the possibility of breaking strategyproofness: if agent $a$ and agent $b$ are both reviewing agent $c$, and agent $a$ is also reviewing agent $b$, agent $b$ may impact the weight given to agent $a$ with their review to agent $c$, thus influencing agent $a$'s role in determining if agent $b$ themself is selected. 
Therefore, we complement the modified selection algorithm (Algorithm~\ref{alg:extpeer}), with a method to assign agents to reviews that does not break strategyproofness, based on the Euler cycle (Figure~\ref{fig:sp_diagram} (Right)). This algorithm for assignment is presented as Algorithm~\ref{alg:eulerAlg}. We separate the two algorithms as they do not depend on each other. Any assignment algorithm that does not create a case in which an agent can both review-with and be reviewed-by the same agent can run Algorithm~\ref{alg:extpeer} and maintain strategyproofness.

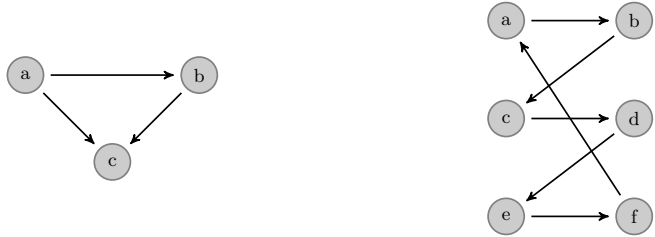
\begin{figure}[h]
 \begin{subfigure}{0.49\linewidth}
 \centering
 \scalebox{0.8}{
 \begin{tikzpicture}
 [auto,
 agent/.style={circle, draw=black!50, fill=black!20,thick, node distance=1 and 1, inner sep = 0mm, minimum size = 6mm},
 valid/.style={->, shorten <=5pt, shorten >=5pt, >=stealth', thick, draw=stasgreen, fill=stasgreen},
 invalid/.style={->, shorten <=5pt, shorten >=5pt, >=stealth', thick, draw=red, fill=red},
 to/.style={->, shorten >=3pt, shorten <=3pt, >=stealth', thick}]
 	\node[agent] (a) 			 {a};
 \node[agent] (c) [below right = of a] {c};
 \node[agent] (b) [above right = of c] {b};
 
 \draw[to]
 (a) edge (b)
 (a) edge (c)
 (b) edge (c)
 ;
 \end{tikzpicture}
}
 \label{triangleFig}
 \end{subfigure}
 \begin{subfigure}{0.49\linewidth}
 \centering
 \scalebox{0.8}{
 \begin{tikzpicture}
 [auto,
 agent/.style={circle, draw=black!50, fill=black!20,thick, node distance=1 and 1.5, inner sep = 0mm, minimum size = 6mm},
 valid/.style={->, shorten <=5pt, shorten >=5pt, >=stealth', thick, draw=stasgreen, fill=stasgreen},
 invalid/.style={->, shorten <=5pt, shorten >=5pt, >=stealth', thick, draw=red, fill=red},
 to/.style={->, shorten >=3pt, shorten <=3pt, >=stealth', thick}]
 	\node[agent] (a) 			 {a};
 \node[agent] (c) [below = of a] {c};
 \node[agent] (e) [below = of c] {e};
 \node[agent] (b) [right = of a] {b};
 \node[agent] (d) [right = of c] {d};
 \node[agent] (f) [right = of e] {f};
 
 \draw[to]
 (a) edge (b)
 (b) edge (c)
 (c) edge (d)
 (d) edge (e)
 (e) edge (f)
 (f) edge (a)
 ;
 \end{tikzpicture}
}
 \label{EulerFig}
 \end{subfigure}
\caption{Algorithm~\ref{alg:eulerAlg} avoids non strategyproof assignments (Left), outputting strategyproof instances (Right).}
 \label{fig:sp_diagram}
 
\end{figure}

Weighting schemes may introduce various strategyproofness violations (as in \cite{Walsh14}), so here we constrain them to be operators of the form $w_{i}:\mathbb{R}^{m} \times \mathbb{R}^{(m-1)\times m}\rightarrow \mathbb{R}$ for each agent $i$. These functions, which in our case will be identical for each agent, associate to each agent $i$'s reviews the values of the other $m-1$ reviews on those same reviewees. Under this condition we can state:

\begin{theorem}
\extpeer is strategyproof.
\end{theorem}
\begin{proof}
As was also shown for \peer \cite{PeerNomination}, for any exogenous set of weights, Algorithm~\ref{alg:extpeer}, is strategyproof -- no agent can influence their own chance of being selected. To show that it is not possible for an agent to influence the weights given to their own reviewers, observe that, for each agent $i$, $w_{i}$ receives only the values of reviews on the same agents that agent $i$ reviews. Thus, it is not possible for an agent $j$, who does not review a paper also reviewed by agent $i$, to change agent $i$'s weight. Consider now an agent $j$ that is reviewing a paper also reviewed by agent $i$. Agent $j$ has a reason to deviate form truthful reporting if changing agent $i$'s weight will benefit the chance that agent $j$ is selected. Since the only use of agent $i$'s weight in the outcome is on agents it reviews, agent $j$ would benefit from altering agent $i$'s weight only if agent $i$ was reviewing agent $j$ themself.

Such an occurrence is prevented by Algorithm~\ref{alg:eulerAlg}. %
Thanks to the initial creation of a bi-partite graph, agent $i$ would review agent $j$ only if $i\in X$ and $j\in Y$ (or vice versa). But then there can be no third agent both of them review: agent $i$ only reviews agents in $Y$, while agent $j$ only reviews agents in $X$ (see Figure~\ref{fig:sp_diagram} (Right)). Thus, the only case where agent $j$ would find it useful to not be truthful cannot happen.
\end{proof}

\begin{algorithm}[t]
\begin{algorithmic}{}
\small
\REQUIRE Set of $n$ agents, review number $m \leq \sfrac{n}{4}$.
\ENSURE Anti-transitive $m$-regular assignment $A$ 

\STATE Initialise $G = (V, E)$ with $V:= [n]$, $E := \emptyset$
\STATE Partition $V$ into $X$ and $Y$ such that $|X| = |Y| = \sfrac{n}{2}$

\STATE \textbackslash\textbackslash Make a $2m$-regular bipartite graph G(V,E)
\FORALL{$x$ in $X$}
    \FORALL{$i$ in $1, ..., 2m$}
        \STATE $y^* \leftarrow {\arg\min}_y \{\deg(y) \mid y \in Y\} $
        \STATE $E \leftarrow E \cup \{\{x, y^*\}\}$
    \ENDFOR
\ENDFOR
\STATE \textbackslash\textbackslash Use Hierholzer’s algorithm to find an Euler cycle
\STATE \textbackslash\textbackslash Euler cycle exists since every node has even degree
\STATE Set euler\_cycle = \emph{Hierholzers\_algorithm}(G)
\STATE Set $A := ([n], E_A)$, $E_A := \emptyset$
\FORALL{$(v_i, v_{i+1})$ in euler\_cycle}
    \STATE $E_A \leftarrow E_A \cup \{(v_i, v_{i+1})\}$
\ENDFOR
\RETURN A
\end{algorithmic}
\caption{\textsc{Euler-Based Assignment}}
\label{alg:eulerAlg}
\end{algorithm}{}

\subsection{A Simple Theoretical Model} \label{sec:theory}

The constructive use of Algorithm~\ref{alg:extpeer} depends on the ability of identifying accurate and inaccurate reviewers, and using this identification to weigh their reviews. Clearly, lacking the knowledge of the ground truth means any identification of accurate/inaccurate agents has to depend on comparing agents to the reviews of the other agents. If all agents were accurate no reweghting of agents would be needed, but as the proportion of accurate agents drops the problem becomes more difficult. Still, when a large majority of them are accurate, the correct opinion is usually the majority, as inaccurate agents will give random rankings.
However, if the number of accurate agents is very low, or other agents are actively malicious, identification becomes impossible, as finding a metric to evaluate the agents against becomes infeasible.

\begin{figure}[t]
 \centering
 \includegraphics[width=0.7\linewidth]{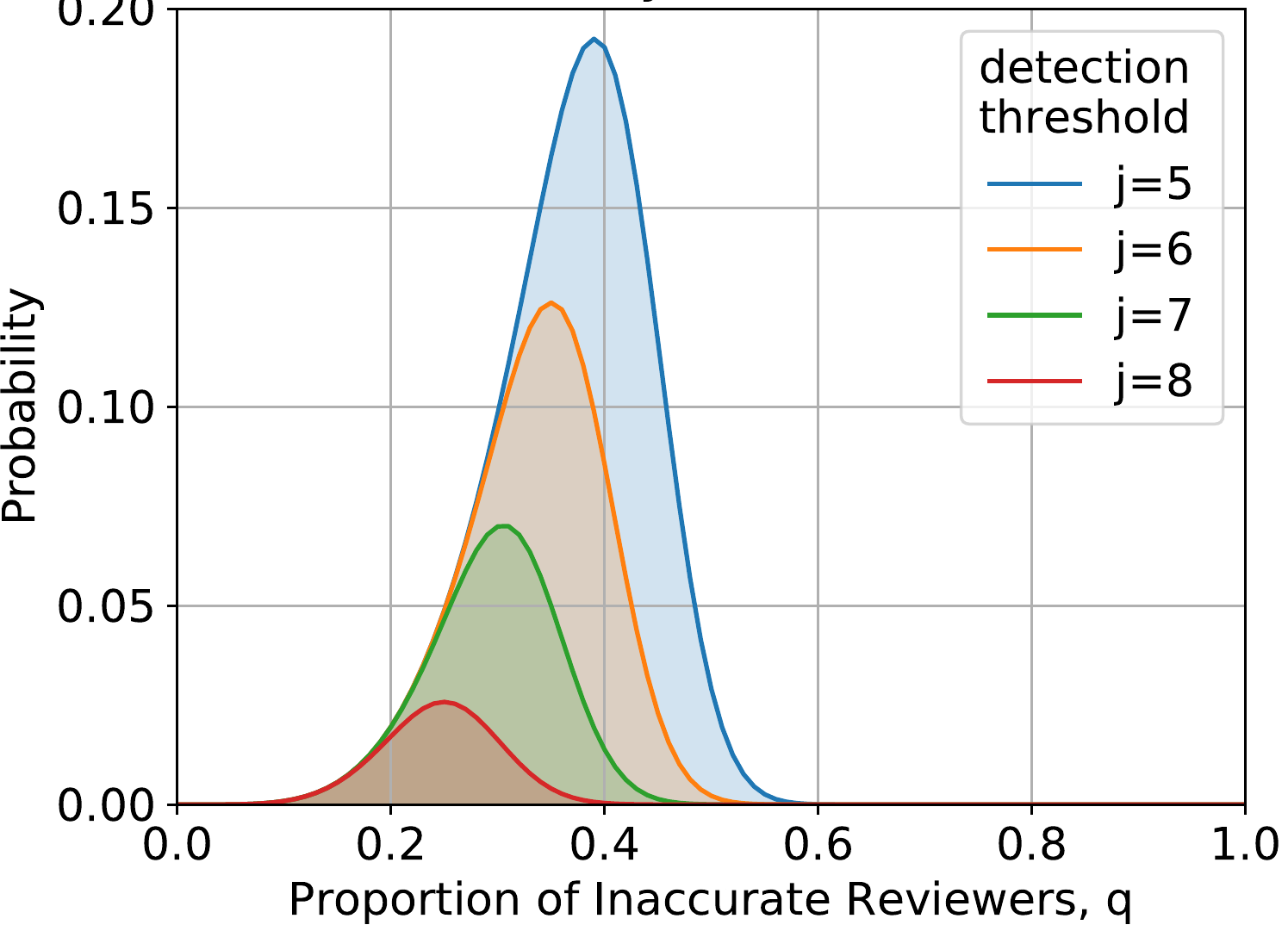}
 \caption{The probability of identifying an inaccurate agent, when $m=9$, and the threshold for identification is $j$.}
 \label{fig:ge_prob}
\end{figure}

To provide some intuition to the conceptual underpinnings of our algorithm we now present a simplified setting, and show how our algorithm -- even with a very simple, conservative, weighting scheme -- is still able to improve over \peer. We start with an $m$-regular assignment, where each agent has one of two types: $\mathcal{A}$, meaning the agent is an accurate reviewer; or $\overline{\mathcal{A}}$, meaning the agent is inaccurate. Recall that in \peer an agent is selected if a majority of their reviewers approve. We show that a very simple dynamic weighting scheme, only relying on knowing how many times an agent has been in a minority, has a good chance of flipping a decision made by $\overline{\mathcal{A}}$ agents to one made by $\mathcal{A}$ agents, improving on \peer. %

\begin{figure*}[htp]
    \begin{subfigure}{0.45\linewidth}
        \centering
        \includegraphics[width=\linewidth]{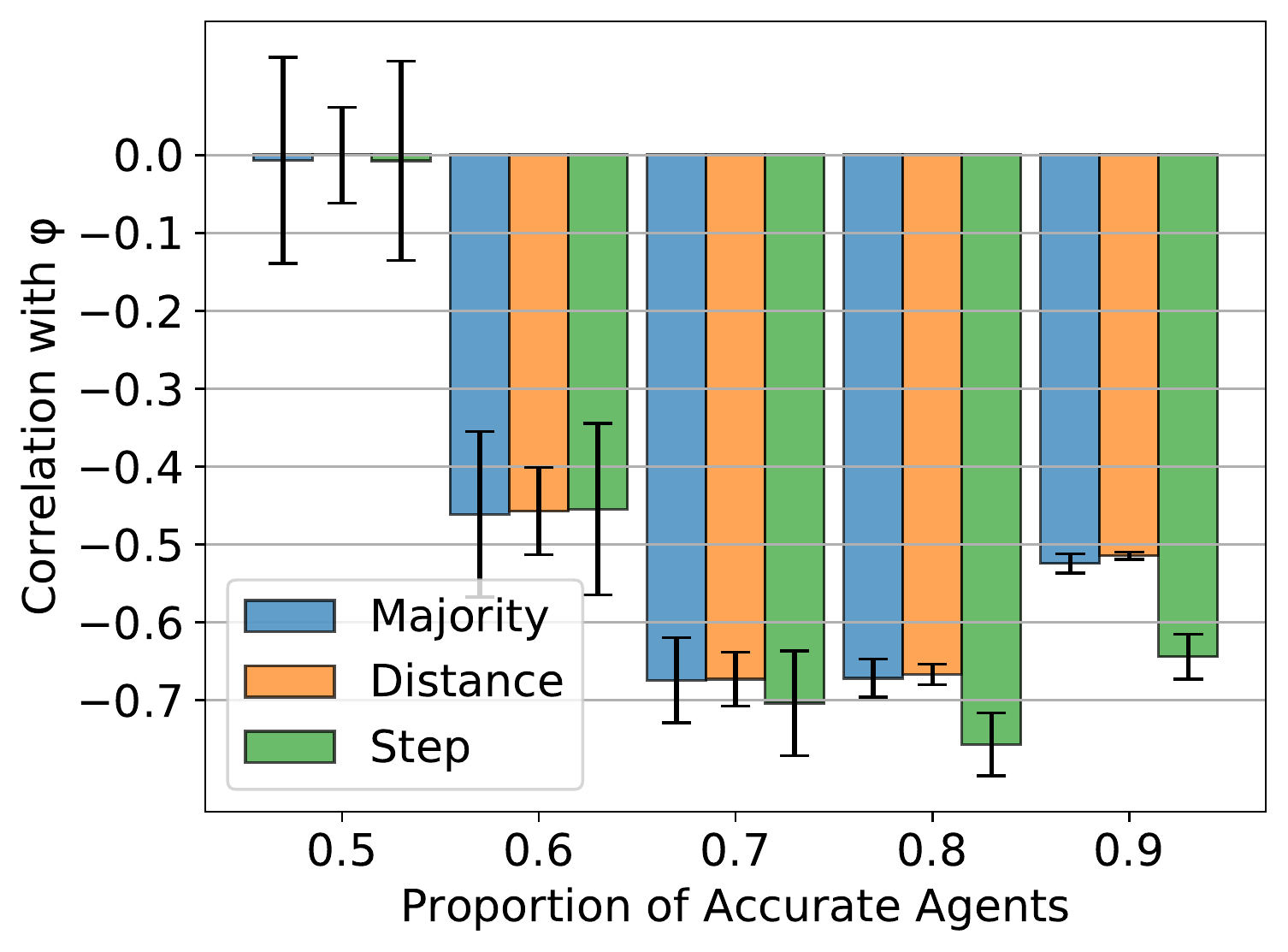}
        \caption{Agents produced from Mallows distribution with $\varphi=0.5$ for the share shown, the rest with $\varphi=1$.}
    \label{fig:correlation-a}
    \end{subfigure}
    \hspace{0.05\linewidth}
    \begin{subfigure}{0.45\linewidth}
        \centering
        \includegraphics[width=\linewidth]{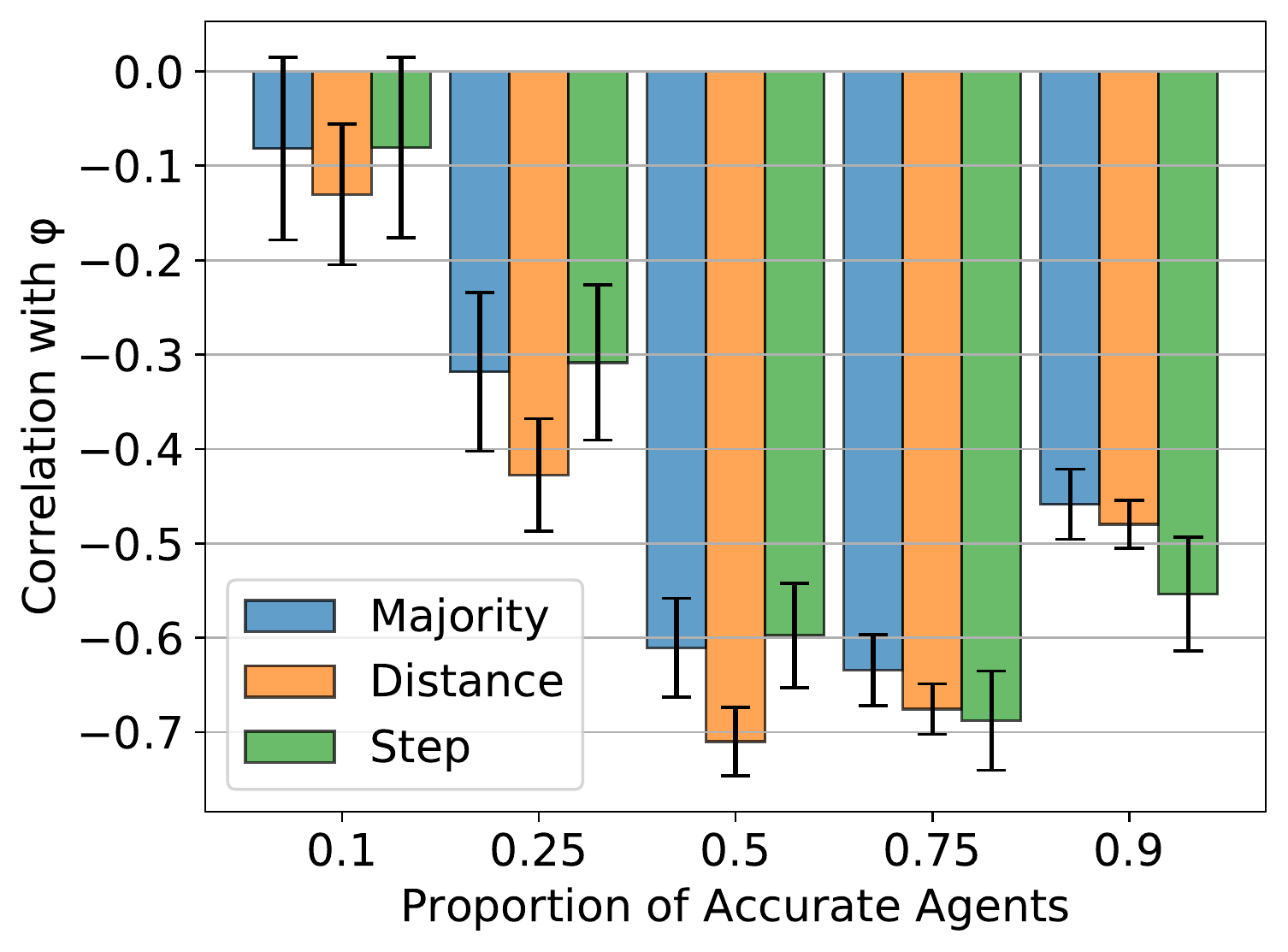}
        \caption{Agents produced from Mallows distribution with $\varphi=0.8$ for the share shown, the rest with $\varphi=1.2$.}
        \label{fig:correlation-b}
    \end{subfigure}
    \caption{Spearman correlation of the weights from the different weighting schemes with the underlying $\varphi$ of each agent. The bars represent the mean and the standard deviation over 1000 simulations of the weights.}
         \label{fig:correlation}
\end{figure*}

Let an $\overline{\mathcal{A}}$-agent be {\em identified as inaccurate} if they held the minority opinion in at least $j$ reviews.
We want to find the probability of the following event: (1) An agent's reviewers have an $\overline{\mathcal{A}}$-majority and (2) Enough of the $\overline{\mathcal{A}}$-agents of that majority are {identified as inaccurate}. 

Given the noise model, let $q$ be the probability for an agent to be of type $\overline{\mathcal{A}}$. %
Then the probability that an agent is reviewed by a majority of $\overline{\mathcal{A}}$ agents and that majority is of size $k$ is: 
\begin{equation*}
\footnotesize
 q_{\overline{\mathcal{A}},k} := \mathbb{P}[\overline{\mathcal{A}}\textrm{-majority of size }k] = {m \choose i}q^i(1-q)^{m-i}
\end{equation*}
where $i = \floor{\sfrac{m}{2}}+ k$. In such a case we would like to identify at least $k$ of the $\overline{\mathcal{A}}$ agents in order to nullify their votes.

We also find the probability that $\overline{\mathcal{A}}$-agents have a majority of any size:
\begin{equation*}
\footnotesize
 q_{\overline{\mathcal{A}}} := \mathbb{P}[\overline{\mathcal{A}}\textrm{-majority}] = \sum_{k=1}^{\lceil m/2 \rceil} q_{\overline{\mathcal{A}},k} = \sum_{i=0}^{\lfloor m/2 \rfloor} {m \choose i} (1-q)^i q^{m-i}
\end{equation*}

Our simple weighting algorithm identifies a $\overline{\mathcal{A}}$-agent if they are in minority for at least $j$ of their other reviewed papers. The probability of this event, $q_{\textrm{det}}$, is given by the cumulative binomial probability, keeping in mind that the probability of $\mathcal{A}$-majority in a set of reviews on an agent is conditioned on the fact that they contain at least one $\overline{\mathcal{A}}$-agent:
\small
\begin{alignat*}{2}
\small
 & q_{\mathcal{A}} &&:=\mathbb{P}[\mathcal{A}\textrm{-majority } | \textrm{ there is at least one }\overline{\mathcal{A}}] \\ 
 & &&=\sum_{i=0}^{\lfloor m/2 \rfloor} {m-1 \choose i} (1-q)^i q^{m-1-i}\\
 \Rightarrow & q_\textrm{det} &&= \sum_{i = j}^{m-1} {m-1 \choose i} (q_{\mathcal{A}}^\prime)^i (1-q_{\mathcal{A}}^\prime)^{m-1-i}
\end{alignat*}
\normalsize
Notice that $\overline{\mathcal{A}}$ may be not just inaccurate but adversarial, in which case we could flip their review and only need to do so for $k$ of them. However, we take the safer approach here, which means we need to detect at least $2k$ $\overline{\mathcal{A}}$-agents to correct the decision. The probability of this correcting event is given by the following expression:
\small
\begin{multline*}
 \mathbb{P}[\textrm{correction event}] = \sum_{k=1}^{\lceil m/2 \rceil} q_{\overline{\mathcal{A}},k} \cdot \left(\sum_{i = 2k}^{\lfloor m/2 \rfloor + k} {\lfloor m/2 \rfloor + k \choose i} q_\textrm{det}^i (1-q_\textrm{det})^{\lfloor m/2 \rfloor + k-i} \right)
\end{multline*}
\normalsize
This produces the desired result, empirically shown in Figure \ref{fig:ge_prob}. As can be seen, for a wide variety of $q$ and $j$, even our very conservative weighting scheme produces a nice probability of improving some reviews. As we shall see with the weighting schemes below, examined by our simulations, even better results can be achieved.

It should be noted that one could produce analogous probabilities of the weighting scheme incorrectly identifying $\mathcal{A}$-agents as $\overline{\mathcal{A}}$-agents. However, for a large enough $j$ (say, $j\geq \sfrac{m}{2}$), and a majority of $\mathcal{A}$ agents (i.e., $q<\sfrac{n}{2}$), this number will always be smaller, i.e., the benefit from the reweghting will be positive.

\begin{figure*}[t]
\centering
    \begin{subfigure}{\linewidth}
        \centering
        \includegraphics[width=\linewidth]{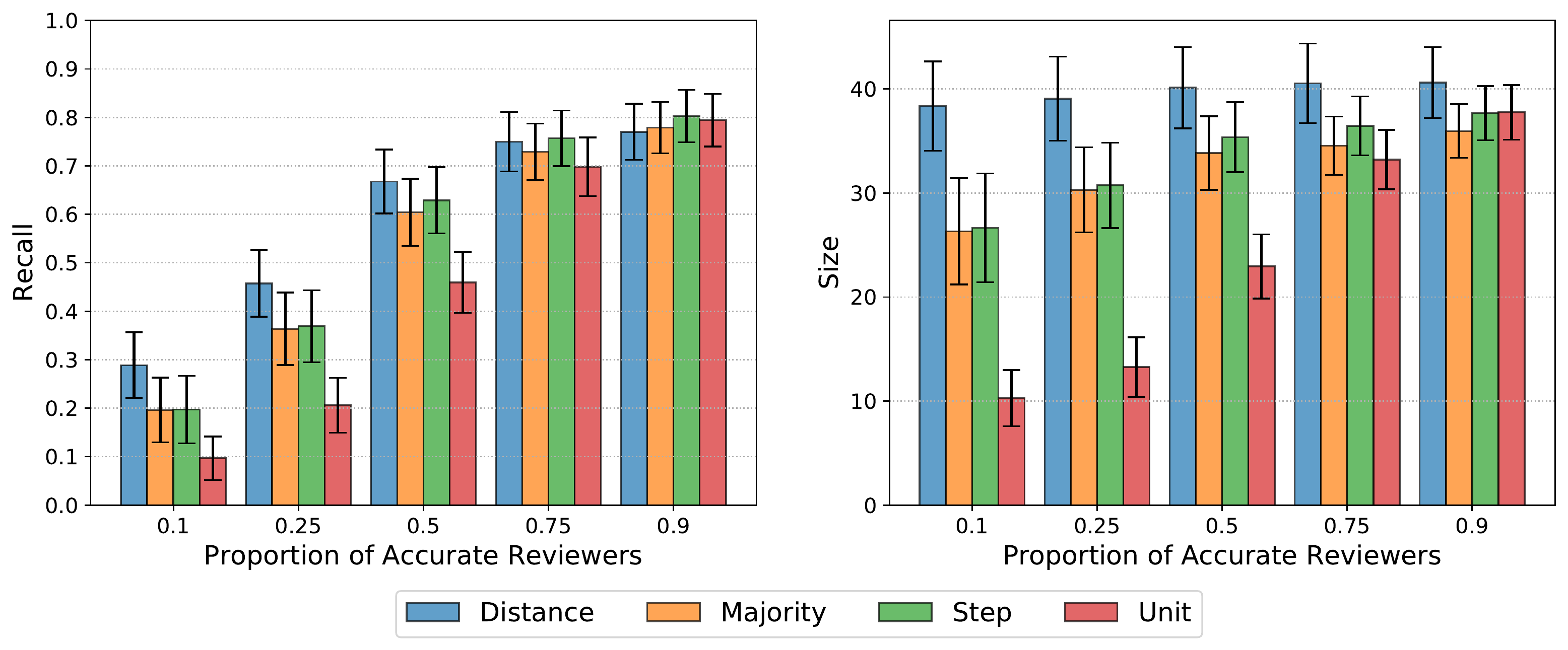}
        \caption{Population's  $\varphi\in\{0.5,1\}$.}
        \label{fig:rnd_results}
    \end{subfigure}

    \begin{subfigure}{\linewidth}
        \centering
        \includegraphics[width=\linewidth]{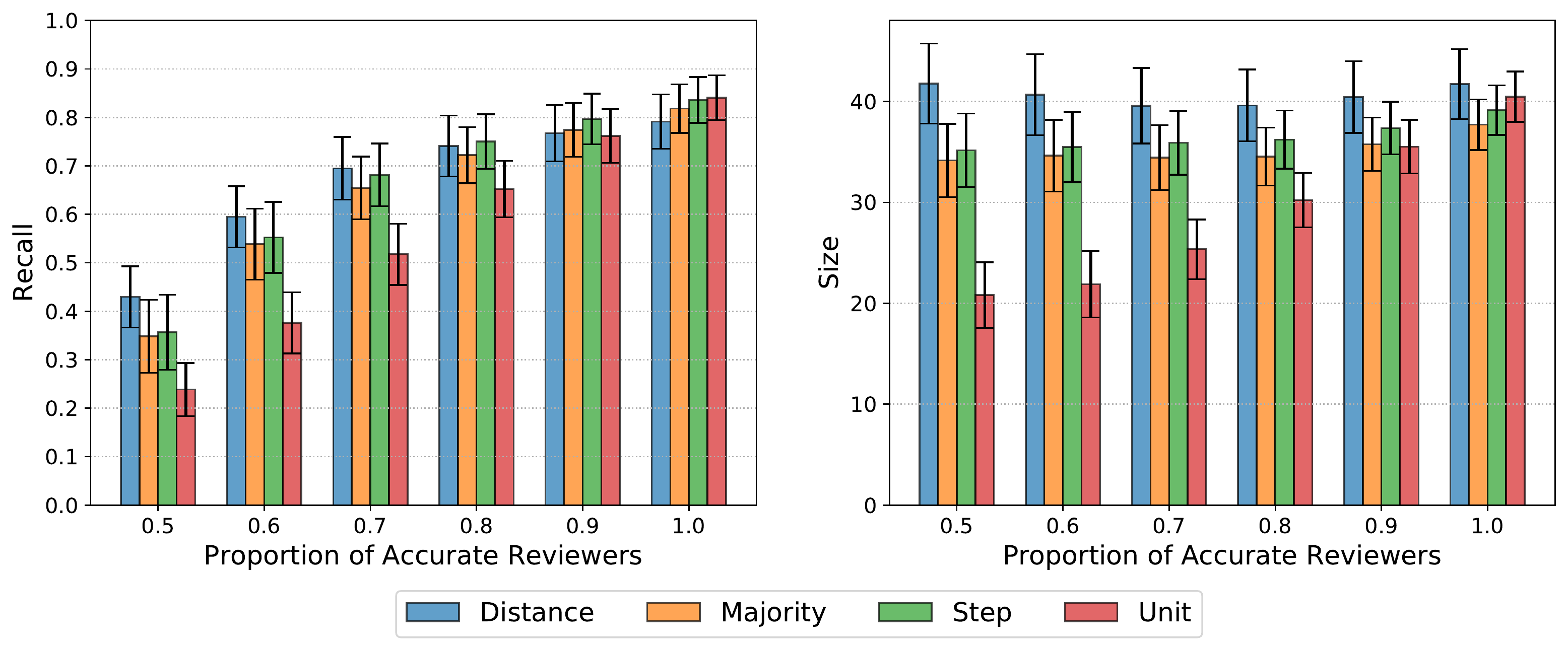}
        \caption{Population's $\varphi\in\{0.8, 1.2\}$.}
        \label{fig:adv_results}
    \end{subfigure}
    \caption{Size and recall with $n=200, m=7, k=40$ with noisy reviewers (a) and adversarial reviewers (b).}
\end{figure*}

\subsection{Weighting Schemes}\label{sec:weightSchemes}

We present three main weighting schemes to evaluate the reliability of the reviewers, each based on the approval vote supplied to the algorithm, maintaining strategyproofness.

\paragraph{Distance.} Distance directly computes the distance between the agent's review and those of other reviewers by simply averaging the individual differences between the reviewers. Formally, define the average distance of reviewer $i$ to other reviewers as $d_i = \frac{1}{m^2} \sum_{j \in A(i)} \sum_{l \in A^{-1}(j)}  |\sigma_i^\textrm{app}(j) - \sigma_l^\textrm{app}(j)| $. Then the distance weight is $w_i^\textrm{dist} = (1 - d_i)^\gamma$, where $\gamma$ is an ``aggression" parameter that exaggerates for better discrimination between the agents.

\paragraph{Majority Errors.} Let an approval of a reviewee be an ``error" by a reviewer if it is a minority opinion, summing fractional nominations. %
Let
\small
\begin{equation*}
    \textrm{maj}_j^\sigma = 
    \begin{cases}
        1,              & \text{if } \sum_{i \in A^{-1}(j)} \sigma_i^\textrm{app}(j) \geq m/2\\
        0,              & \text{otherwise}
    \end{cases}
\end{equation*}
\normalsize
be the majority for agent $j$ in profile $\sigma$. Then define the number of errors of reviewer $i$ to be 
$\textrm{err}_i^\sigma = \sum_{j \in A(i)} \mathds{1}_{\sigma_i^\textrm{app}(j) \neq \textrm{maj}_j^\sigma}$. Lastly, the weight is defined as $w_i^\textrm{majerr} = 1 - \delta(\textrm{err}_i^\sigma/m)$, where $\delta$ is an aggression parameter.

\paragraph{Step.} Step applies a step function to the error rate $\textrm{err}_i^\sigma/m$ defined above. We choose two thresholds, $t_1$ and $t_2$ such that if the error rate reaches $t_1$, we reduce the weight of the reviewer to $0.5$; if the error rate reaches $t_2$, we reduce the weight to 0. Additionally, we scale each threshold by the nomination quota as it plays a bigger role in error detection than just the size of the review pool, $m$. Formally,
\begin{equation*}
    w_i = \begin{cases}
    1,              & \text{if } \textrm{err}_i^\sigma/\frac{n}{k} < t_1\\
    0.5,            & \text{if } t_1 \leq \textrm{err}_i^\sigma/\frac{n}{k} < t_2\\
    0,              & \text{otherwise}.
    \end{cases}
\end{equation*}
\normalsize

\section{Setup and Performance Results}\label{sec:experiments}

We test \extpeer, using our weighting schemes, against \peer, following the testing framework used in \cite{PeerNomination} and \cite{Aziz2019}. However, our setup differs in a few significant ways. In contrast to the previous experiments, we do not include any partition-based mechanisms, thus allowing us to use the assignment-generating procedure given in Algorithm \ref{alg:eulerAlg} without worrying about forming a given number of clusters (the procedure does generate a 2-clustering).

To simulate the presence of noise, we use Mallows model, as described in Section \ref{sec:mallows}. We assume that the population consists of accurate and inaccurate reviewers, which we represent by using two distinct values of the dispersion parameter $\varphi$. In Section \ref{sec:random} we assume that the accurate reviewers have $\varphi=0.5$ and the inaccurate reviewers have $\varphi=1$, i.e., their reviews are random. We vary the proportion of accurate reviewers between 0.1 and 0.9. Also recall that in Section \ref{sec:mallows} we extended the definition of the dispersion parameter to range between $0$ and $2$. In Section \ref{sec:adversarial}, we set $\varphi=0.8$ for the accurate reviewers and $\varphi=1.2$ for the adversarial reviewers, i.e., most of their reviews will contradict the ground truth. We vary the proportion of accurate reviewers between $0.5$ and $1$. Note that as the proportion of accurate reviewers drops to 0.5 and lower, it becomes impossible to recover the ground truth.

Figure \ref{fig:correlation} demonstrates that our metrics strongly correlate with the underlying $\varphi$ in both setups, demonstrating their effectiveness in singling out accurate reviewers.

\subsection{Inaccurate Reviewers} \label{sec:random}

\sloppy Figure \ref{fig:rnd_results} shows %
that a high (90\%) proportion of accurate reviewers results in %
barely any improvement over Unit (original \peer). This indicates that the weighting schemes do not overfit in the search for non-existing noise. %
 However, as the proportion of random reviewers rises, all weighting schemes outperform Unit.%

We also see that the weighting schemes are indeed much better at keeping the output size close to the desired $k$ with Distance keeping the output size consistent across all levels of noise. In addition, the much greater recall of Distance and other weighting schemes indicates the additional selected agents are usually the deserving ones. 
Between the weighting schemes, Distance manages to gain more and more advantage as the noise levels increase as it is the most fine-grained and aggressive one. This allows it to both identify the inaccurate reviewers and maintain consistent output size.%

There are similar patterns for other settings of parameters. As $k$ increases, the performance of all algorithms improves as not only can they be less selective, but the reviewers provide more data. The advantage over Unit, however, decreases as \peer is capable of taking advantage of this as well. Increasing $m$, on the other hand, primarily benefits the weighting schemes as they can calculate the weights more accurately, increasing their advantage over Unit.

\subsection{Adversarial Reviewers} \label{sec:adversarial}

We see a similar story with the adversarial model. Again, at low noise levels, \peer, as expected, performs well while the weighting schemes' slight drop in performance can be attributed to overfitting. However, as the proportion of adversarial reviewers rises, we can observe a notable increase in performance for the weighting schemes. Even when half of the population is adversarial, Distance impressively achieves recall of over 40\% as compared to the theoretical max of 50\%.

\section{Conclusion}

In this paper we propose a novel strategyproof peer selection algorithm -- \extpeer, which weighs reviewers based on their perceived accuracy. The basis for this reweghting is the observation that in most cases, one's reviews are correlated in quality. %
We develop several weighting methods, showing that even straightforward ones can reach high quality outcomes, with high level of noise. Our algorithm is constructed in a modular way, allowing for a variety of weighting and evaluation methods to be developed for particular settings or noise models. This modularity allows for multiple directions of future development, including developing additional weighting; matching weighting schemes to noise models; and optimising them using various techniques.

\newpage

\footnotesize
\bibliographystyle{splncs04}
\bibliography{peer}

\end{document}